\documentclass{amsart}

\usepackage[title]{appendix}
\usepackage{graphicx}
\usepackage{amsfonts}
\usepackage{url}
\usepackage{amscd}
\usepackage{amssymb}
\newtheorem{algor}{\bf{Algorithm}}[section]
\usepackage[algosection,lined,boxed,commentsnumbered,ruled,vlined]{algorithm2e}
\usepackage[hidelinks]{hyperref}

\newtheorem{remark}{Remark}[section]

\newtheorem{theorem}{Theorem}[section]
\newtheorem{lemma}[theorem]{Lemma}

\theoremstyle{remark}

\numberwithin{equation}{section}

\usepackage{indentfirst}
\usepackage[greek,english]{babel}
\usepackage[utf8]{inputenc}
\usepackage{color,amsmath,amssymb,graphicx}
\usepackage{tikz-cd}
\usepackage{amsmath,amssymb}
\usepackage{enumerate}
\usepackage{dsfont}
\usepackage{amsthm}
\usepackage{float}

\usepackage[algosection,lined,boxed,commentsnumbered,ruled,vlined]{algorithm2e}

\usepackage{mathtools}

\newfont{\Bdd}{msbm10 scaled\magstep1}
\newfont{\footnotesizeBdd}{msbm8 scaled\magstep1}

\author{M. Adamoudis}
\address{Department of Mathematics, 
Aristotle University of Thessaloniki, 54 124, Thessaloniki, Greece}
\email{aamarios@math.auth.gr}

\author{K. A. Draziotis}
\address{Department of Informatics, Aristotle University of Thessaloniki, 54 124, Thessaloniki, Greece}
\email{drazioti@csd.auth.gr}

\author{D. Poulakis}
\address{Department of Mathematics, 
Aristotle University of Thessaloniki, 54 124, Thessaloniki, Greece}
\email{poulakis@math.auth.gr}

\begin{document}

\title[Computation of the private key]{Attacking (EC)DSA scheme with ephemeral keys sharing specific bits}

\keywords{Digital Signature Algorithm; Elliptic Curve Digital Signature Algorithm; Lattice; LLL; Kannan's Enumeration Algorithm}
\subjclass[2010]{94A60.}
\maketitle	

\begin{abstract} 
In this paper, we present a deterministic  attack
 on (EC)DSA signature scheme, providing
that several signatures are known such that the corresponding 
ephemeral keys share a certain amount of bits without knowing
their value. By eliminating the shared blocks of bits between 
the ephemeral keys, we get a lattice of dimension equal to the number of signatures having a vector containing the private key. We compute
an upper bound for the distance of this vector from a target vector, and next,
using Kannan's enumeration algorithm, we determine it and hence the secret key.
The attack can be made highly efficient by appropriately selecting  
the number of shared bits and the number of signatures.
\end{abstract}

\section{Introduction - Statement of results}
In August 1991, the U.S. government's National Institute of
Standards and Technology (NIST) proposed an algorithm for digital
signatures. The algorithm is known as DSA, for Digital Signature
Algorithm \cite{National, Menezes, Koblitz2}. It  is an efficient
variant of the ElGamal digital signature scheme \cite{ElGamal}
intended for use in electronic mail, electronic funds transfer,
electronic data interchange, software distribution, data storage,
and other applications which require data  integrity assurance and
data authentication. In 1998, an elliptic curve analogue called
Elliptic Curve Digital Signature Algorithm (ECDSA) was proposed
and standardized  \cite{Johnson, Koblitz,Koblitz2}.

\subsection{The (EC)DSA Signature Scheme}
First, we recall the  DSA schemes. The
signer selects a prime $p$  of size between 1024 and 3072 bits with
 increments of 1024, as recommended in FIPS 186-3 \cite[page 15]{fips}. 
 Also, he selects a prime $q$ of size 160, 224 or 256 bits, with $ q|p-1$
  and  a generator $g$ of the
unique order $q$ subgroup $G$ of the multiplicative group $\mathbb{F}_p^*$
of the prime finite field $\mathbb{F}_p$. Furthermore, he
selects a randomly $a \in \{1,\ldots,q-1\}$ and computes $R = g^a \, \bmod\,  p$. 
The public key of the signer is  $(p,q,g,R)$ and his private
key $a$. He also publishes a hash function
 $h : \{0,1\}^* \rightarrow \{0,\ldots,q-1\}$. 
 To sign a message $m\in \{0,1\}^*$, he selects  randomly
$k \in \{1,\ldots,q-1\}$ which is the ephemeral key, and
 computes
$r = (g^k \ \bmod \ p) \ \bmod \ q$ and 
$s = k^{-1}(h(m)+ar) \, \bmod \, q$. 
The signature of $m$ is $(r,s)$. The  signature is accepted as
 valid if and only if the following holds:
$$r = ((g^{s^{-1}h(m)\bmod\ q}R^{s^{-1}r \bmod\, q})\, \bmod\, p) \,
 \bmod \, q.$$
Next, let us recall the  ECDSA scheme. The signer selects an elliptic curve 
$E$ over $\mathbb{F}_p$,   a point $P\in E(\mathbb{F}_p)$ with  order a prime  
$q$  of size at least 160 bits. 
Following FIPS 186-3, the prime $p$ must belongs to the set
$\{160,224,256,512\}.$ Further, he chooses randomly  
$a \in \{1,\ldots,q-1\}$ and computes $Q = aP$. 
 Finally, he publishes a hash
 function  $h : \{0,1\}^* \rightarrow \{0,\ldots,q-1\}$.
The public
key of the signer is $(E,p,q,P,Q)$ and his private key $a$.
To sign a message $m$, he selects randomly
 $k \in \{1,\ldots,q-1\}$ which is the ephemeral key and computes
$kP = (x,y)$ (where $x$ and $y$ are regarded as integers between 0 and 
$p-1$).
He computes
$r = x \ \bmod \ q$  and 
$s = k^{-1}(h(m)+ar) \ \bmod \ q$.
The signature of $m$ is $(r,s)$. The verifier  computes
$$u_1 = s^{-1}h(m) \ \bmod \ q, \ \  u_2 = s^{-1}r \ \bmod \ q, \ \
u_1P+u_2Q = (x_0,y_0).$$
He accepts the signature if and only if $r = x_0 \, \bmod  \, q$. 

\subsection{Previous Results}

Researchers have explored various attacks on DSA schemes by analyzing the signature equation $s= k^{-1}(h(m)+ar) \ {\rm mod} \ q$ and using lattice reduction techniques such as LLL and CVP algorithms. One study focused on the use of a linear congruential pseudorandom number generator (LCG) for generating random numbers in DSA  \cite{Bellare}, showing that combining the DSA signature equations with LCG generation equations can lead to a system of equations that provide the secret key. To recover the secret key, several heuristic attacks have been proposed \cite{Howgrave} in another study, which assume the revelation of a small fraction of the corresponding nonce $k$. However, these attacks are based on heuristic assumptions, making it difficult to make precise statements on their theoretical behavior. 

The first rigorous lattice attack on (EC)DSA was presented in \cite{Nguyen}. The authors successfully decreased the security of (EC)DSA to a Hidden Number Problem (HNP), which can then be further reduced to an approximation Closest Vector Problem (CVP) for a specific lattice.  The signer's secret key $a$ can be computed using this reduction in polynomial time. The attack was also adapted to the case of ECDSA, as described in \cite{Nguyen2}.

The paper \cite{Blake} describes an attack on DSA schemes that uses the LLL reduction method and requires one message. By computing two short vectors of a three-dimensional lattice, the attack derives two intersecting lines in $(a, k)$, provided that $a$ and $k$ are sufficiently small, and the second shortest vector is sufficiently short. If two messages are available, the same attack can be applied to derive a linear congruence relating to the corresponding ephemeral keys.

The papers \cite{Poulakis} and \cite{Draziotis} describe attacks on DSA schemes using the LLL algorithm and  one or two messages. In \cite{Poulakis},  the combination of LLL   with  algorithms for finding integral points of two classes of conics gives $a$, provided that
 at least one of the sets $\{a,k^{-1}\, \bmod\,q\}$, $\{k,a^{-1}\, \bmod\,q\}$, $\{a^{-1}\,\bmod\,q,k^{-1}\,\bmod\,q\}$ is sufficiently small. 
In \cite{Draziotis}, the Lagrange Reduction algorithm is applied
on a 2-dimensional lattice  defined by a signed message, and provides
 two straight lines intersecting at $(a, k)$. Similar attacks can be applied to the pairs $(k^{-1}\,\bmod\,q, k^{-1}a\, \bmod\, q)$ and $(a^{-1}\,\bmod\, q,a^{-1}k\,\bmod\,q)$. If two signed messages are available, the above two attacks can be applied to the equation relating the two ephemeral keys.

The article \cite{Draziotis2} presents an attack using Coppersmith's method to compute the secret key $a$. The attack works when $a$ and $k$ satisfy a specific inequality, and in this case, the secret key $a$ can be efficiently computed.

The article \cite{Poulakis1} describes an attack that involves constructing a system of linear congruences using signed messages. This system has at most one unique solution below a certain bound, which can be computed efficiently. Thus, if the length of a vector containing the secret and ephemeral keys of a signed message is quite small, the secret key can be computed using the above system. The article \cite{marios} presents an improved version of this attack.

In \cite{Mulder1, Mulder2}, the proposed attacks 
take advantage using of the bits in the ephemeral key and the Fast Fourier Transform.

In \cite{Sun}, it is shown that, using lattice reduction under some
heuristic assumptions, that  partial
information about the nonces of multiple signatures can lead to recovery of the 
full private key. The  original approach to doing so is based  on discrete
Fourier analysis techniques \cite{Bleichenbacher, Aranha}.  

 A very important issue is the
 attacks on cryptosystems based on the malicious modification of memory registers.  These  attacks may affect the randomness of the secret parameters, 
and so, to force   certain bits of the ephemeral
key to be  equal, without  their values being known. In \cite{Leadbitter},
it is discussed  how such attacks could occur in a real-life scenario.
Following the line of research from \cite{Leadbitter}, the authors of \cite{Faugere} focus on an attack scenario where ephemeral keys share specific bits, such as the least significant bits (LSB) and/or most significant bits (MSB), either within multiple blocks. 
By eliminating the shared blocks
 of bits between the ephemeral
keys, a lattice of dimension equal to the number of signatures is provided, which 
contains a quite short vector with components that reveal the secret key.
Then, the LLL algorithm is used for the computation of this vector.  
Note that these attacks are based on heuristic assumptions. 
Later, in \cite{Gomez}, the authors further improved upon the attack proposed in \cite{Faugere} by providing a probabilistic attack with a success probability approaching $1$ when the pair $(\delta,n)$ is appropriately selected, where $n$ represents the number of signatures, and $\delta$ represents the number of shared bits in the ephemeral keys. This attack relies on a mild assumption regarding the hash function used in (EC)DSA.

\subsection{Our Contribution}

Our study builds on the research presented in \cite{Faugere, Gomez}, and we present a deterministic attack that, although not always polynomial in complexity, proves to be highly efficient in practical scenarios. Instead of using methods like LLL, approximate, or exact CVP, which were employed in previous attacks, we use enumeration on a suitable lattice to find lattice vectors that are close to a specific target vector. From these solutions, we can readily extract the secret key to the system.

It is important to highlight that the attacks presented in \cite{Faugere} rely on heuristics assumptions that aim to force the presence of a vector containing the private key as a solution to the Shortest Vector Problem (SVP) in a relatively large lattice. In \cite{Gomez}, the authors provide a probabilistic approach to \cite{Faugere}, where an assumption for the hash function is made and the attack is modelled as a Closest Vector Problem (CVP). Due to the computational complexity of finding such a vector using a deterministic algorithm, an approximation algorithm can be used instead. 

Our approach takes a different path. We calculate a bound for the distance between the vector of the lattice containing the private key and a target vector. Then, we leverage Kannan's enumeration algorithm to determine this vector and, consequently, extract the secret key. Our experiments demonstrate that the attack can be made highly efficient by appropriately selecting values for $\delta$ and $n$. Finally, we improve the results provided in \cite{Gomez}.

\subsection{Our results}
In the subsequent Theorem, we apply the framework suggested by \cite{Gomez, Faugere, Leadbitter}, which presupposes that we have access to a collection of signed messages with ephemeral keys that are shorter than $q$. These messages have some of their most and least significant bits in common, with a total of $\delta$ bits shared.

\begin{theorem} \label{theorem}
Suppose we have a $(EC)DSA$ scheme
with a binary length $\ell$ prime number $q$ and secret key $a.$ Let $m_j$ $(j=0,\ldots,n)$  be
messages signed with this scheme, $(r_j,s_j)$ their signatures, and
 $k_j = \sum_{i=1}^{\ell} k_{j,i} 2^{\ell-i}$ (where $k_{j,i}\in \{0,1\}$) are 
 the  corresponding ephemeral keys, respectively.
  Set $A_j =  -r_js_j^{-1} \bmod\ q$.
 Suppose that $0< k_j < q$ $(j=0,\ldots,n)$, and there are integers $\delta >0$  and 
 $0 \leq \delta_L\le \delta$ such that the
 following conditions hold:
 \begin{enumerate}
 \item   $k_{0,i+1} = \cdots = k_{n,i+1}$  $(i=1,\ldots,\delta-\delta_L,\ell-\delta_L,
\ldots,\ell-1)$.
\item  For $i = 0,\ldots,n$, set 
$C_{i,j} = (A_{j-1} -A_i) 2^{-\delta_L} \ \bmod\ q$, $(j=1,\ldots,i)$, and
$C_{i,j} = (A_j -A_i) 2^{-\delta_L} \ \bmod\ q$
$(j=i+1,\ldots,n)$.
The shortest vector of the lattice $\mathcal{L}_i$  spanned by the vectors
$$(2^{\delta+1}q,0,\ldots, 0),\ldots,
(0,\ldots, 0, 2^{\delta+1}q , 0),
(2^{\delta+1}C_{i,1}, \ldots, 2^{\delta+1}C_{i,n}, 1)$$ 
has length 
 $$> \frac{1}{2} \, (2^{\delta+1}q)^{\frac{n}{n+1}}.$$
\end{enumerate}
Then, the secret key  $a$ can be computed  in 
 $${\mathcal{O}}(2^{\ell-\delta n+2n}\,n\,  ( (n\ell)^c \,2^{{\mathcal{O}}(n)} 
 +\ell^4 2^n (n+1)^\frac{n+1}{2}))$$
 bit operations, for some $c > 0$.
\end{theorem}

\begin{remark}{\rm 
By the Gaussian heuristic \cite[Section 6.5.3]{Hoffstein} the length of the vectors of the lattice 
$\mathcal{L}$ is $> q^{n/(n+1)}$.
Thus, the hypothesis (2) of Theorem \ref{theorem} will very often be
satisfied.}
\end{remark}
\begin{remark}{\rm 
In the  above complexity estimate, if $\ell\leq \delta n$,  then 
 the time complexity is polynomial in $\ell$.}
\end{remark}\ \\
{\bf Roadmap}. The paper is structured as follows: 
Section 2 presents an auxiliary lemma that will prove  crucial in the proof of Theorem \ref{theorem}. 
Section 3 is dedicated to the proof of Theorem \ref{theorem}, providing a detailed explanation and justification. 
In Section 4, an attack on (EC)DSA, derived from Theorem \ref{theorem}, is presented. Additionally, several experiments are conducted to illustrate the effectiveness of the attack.
Finally, Section 5 concludes the paper, summarizing the main findings and discussing potential avenues for future research.

\section{Lattices}

Let $\mathcal{B} = \{{\bf b}_1, \ldots, {\bf b}_n\}\subset \mbox{\Bdd Z}^n$ be a basis of $\mbox{\Bdd R}^n$.
 A $n$-{\em dimensional lattice} spanned by $\mathcal{B}$ is the set
$$\mathcal{L} = \{z_1{\bf b}_1+\cdots +z_n{\bf b}_n/\ z_1,\ldots,z_n \in \mbox{\Bdd Z}\}.$$

Recall that the scalar product of two vectors  $\mathbf{u} = (u_1,\ldots,u_n)$ 
and $\mathbf{v} = (v_1,\ldots,v_n)$ in $\mathbb{R}$ is the quantity $\langle \mathbf{u},\mathbf{v} \rangle = u_1v_1+\cdots + u_nv_n$, and 
the {\em Euclidean norm} of a vector ${\bf v} = (v_1,\ldots,v_n) \in \mbox{\Bdd R}^n$ 
the quantity 
$$\|\mathbf{v}\| = \langle \mathbf{v},\mathbf{v} \rangle^{1/2} =
(v_1^2+\cdots +v_n^2)^{1/2}.$$

The Gram-Schmidt orthogonalisation (GSO) of the basis  $\mathcal{B}$
 is the orthogonal family 
 $\{\mathbf{b}_1^{\star},\ldots,\mathbf{b}_n^{\star}\}$
  defined as follows:
 $$\mathbf{b}_i^{\star} =
  \mathbf{b}_i-\sum_{j=0}^{i-1}\mu_{i,j}\mathbf{b}_j^{\star}, \ \ \ 
 {\rm with}\ \ \ 
 \mu_{i,j} = \frac{\langle \mathbf{b}_i,\mathbf{b}_j^{\star}\rangle}{\|\mathbf{b}_j^{\star}\|^2}   \ \  \ (j= 0,\ldots,i-1).$$

Let   $L$ be a lattice. If $K$ is a convex body
in $\mbox{\Bdd R}^{n+1}$ symmetric about the origin, we denote by 
 $\lambda_i(K,L)$ $(i=1,\ldots,n+1) $
 the $i$th successive minimum of $K$ with respect to $L$
 which it is defined as follows
$$\lambda_i(K, L) = \inf\{ \lambda > 0/\ (\lambda K) \cap L\ {\rm contains}\ i 
\ {\rm linearly\ independent\ points}\}.$$ 
Further, we denote by $s(L)$ the length of a shortest vector in $L$.

\begin{lemma} \label{sphere}
Let $B_{\mathbf{v}}(R)$ be the closest ball of center $\mathbf{v}$ and 
radius $R$ in $\mathbb{R}^{n+1}$ and $L$ a lattice. Then,we have:
$$|B_{\mathbf{v}}(R)\cap L | < \left( \frac{2R}{s(L)}+1\right)^{n+1}.$$
\end{lemma}
\begin{proof} Set
$$\mathcal{D}_\mathbf{v}(R) =
 \{\mathbf{x}-\mathbf{y}/\ \mathbf{x},\mathbf{y}\in B_{\mathbf{v}}(R)\}.$$
Then, $\mathcal{D}_{\mathbf{v}}(R)$ is a convex body, symmetric about the 
origin.
Then \cite{Malikiosis} implies:
\begin{equation} \label{ineq. 4.2.1}
|B_{\mathbf{v}}(R)\cap L | < 
\prod_{i=1}^{n+1} \left(\frac{1}{\lambda_i(\mathcal{D}_{\mathbf{v}}(R),L)}+1\right).
\end{equation}
Let $\mathbf{x},\mathbf{y}\in B_{\mathbf{v}}(R)$. Then, we have:
$$\|\mathbf{x}-\mathbf{y}\| \leq \|\mathbf{x}-\mathbf{v}\|+
\|\mathbf{v}-\mathbf{y}\| \leq 2R.$$
It follows that $\mathcal{D}_{\mathbf{v}}(R)\subseteq B_{\mathbf{0}}(2R)$,
and so we deduce
\begin{equation}\label{ineq. 4.2.2}
 \lambda_1(B_{\mathbf{0}}(2R),L) \leq \lambda_i(\mathcal{D}_{\mathbf{v}}(R),L) \ \  (i=1,\ldots,n).
\end{equation}
Further, we have
\begin{equation}\label{ineq. 4.2.3}
\lambda_1(B_{\mathbf{0}}(2R),L) \geq  s(L)/2R. 
\end{equation}
Combining the inequalities (\ref{ineq. 4.2.1}), (\ref{ineq. 4.2.2})
and (\ref{ineq. 4.2.3}), we obtain:
$$|B_{\mathbf{v}}(R)\cap L | < \left( \frac{2R}{s(L)}+1\right)^{n+1}.$$
\end{proof}

\section{Proof of Theorem 1.1} 
Let $a$ be the secret key and $k_j, \ j = 0,\ldots,n$ the ephemeral keys. We put   $A_j =  -r_js_j^{-1} \bmod\ q$  and $B_j = -h(m_j) s_j^{-1} \bmod\ q $ for $j = 0,\ldots,n.$ 
The signing equation for (EC)DSA provides that,
\begin{equation}\label{signing_equation}
k_j+A_j a +B_j \equiv 0 \ (\bmod\ q)\ \ \ (j=0,\dots,n).
\end{equation}
Suppose first that $k_0 = \min\{k_0,\ldots,k_n\}$.  
We set $\delta_M=\delta-\delta_L.$ From the hypothesis of the Theorem we get 
$$z_j=k_j-k_0=\varepsilon 2^{\ell-\delta_M-1}+\cdots+\varepsilon' 2^{\delta_L},$$
for some $\varepsilon, \varepsilon'\in\{0,1\}.$
Since $z_j>0$ we get $0<z_j<2^{\ell-\delta_M}$ and there exists positive integer $z_j'$ such that
$z_j = 2^{\delta_L}z^{\prime}_j$
Furthermore, we set
$C_j = (A_j-A_0)2^{-\delta_L}\ \bmod \ q$ and 
$D_j = (B_j-B_0)2^{-\delta_L}\ \bmod \ q.$
From (\ref{signing_equation}) we have the congruences:
$$z_j^{\prime}+C_j a +D_j \equiv 0 \ (\bmod\ q) \ \ \ (j=1,\ldots,n).$$
Since $z_j^{\prime}$ is positive, there is a positive integer  $c_j$ such that  
$$-C_ja-D_j+c_jq= z_j^{\prime}.$$
 Thus, we obtain:
$$0 < c_jq-C_j a-D_j < 2^{\ell-\delta}.$$
It follows that
$$-2^{\ell-\delta-1} < c_jq-C_j a-D_j-2^{\ell-\delta-1} < 2^{\ell-\delta-1},$$
whence we get
$$0 < |c_jq-C_j a-D_j-2^{\ell-\delta-1}| < 2^{\ell-\delta-1}.$$
Therefore, we have:
\begin{equation} \label{inequality}
0 < |c_jq2^{\delta+1} -C_j2^{\delta+1} a-D_j2^{\delta+1}-2^{\ell}|
 < 2^{\ell}.
 \end{equation}

We consider the lattice $\mathcal{L}$ spanned by the rows of the matrix
\[
\mathcal{J} =   \left( \begin{array}{ccccccc}
		2^{\delta+1}q & 0 & 0 & \ldots  & 0 & 0 \\
		0 & 2^{\delta+1}q & 0 & \ldots  & 0 & 0 \\
		0 & 0 & 2^{\delta+1}q & \ldots  & 0 & 0 \\
		\vdots & \vdots & \vdots & \ddots & \vdots & \vdots \\
		0 &    0  & 0   &  \ldots & 2^{\delta+1}q  & 0\\
		2^{\delta+1}C_1 & 2^{\delta+1}C_2 & 2^{\delta+1}C_3  & \ldots & 2^{\delta+1}C_n & 1
\end{array}
\right).
\]
The vectors of the lattice ${\mathcal{L}}$ are of the form
$$(2^{\delta+1}(qx_1+x_{n+1}C_1),2^{\delta+1}(qx_2+x_{n+1}C_2),\dots,2^{\delta+1}(qx_n+x_{n+1}C_n),x_{n+1}),$$
for some integers $x_1,\ldots,x_{n+1}.$ By setting 
$(x_1,\ldots,x_{n+1})=(c_1,\ldots,c_n,-a)$, we get the lattice vector
$$\mathbf{u} = 
(2^{\delta+1}(c_1q-C_1a),\ldots,2^{\delta+1}(c_nq-C_na),-a).$$
Further we consider the vector in the span of ${\mathcal{L}}$,
$$ \mathbf{v} = (D_12^{\delta+1}+2^{\ell},\ldots,2^{\delta+1}D_n+2^{\ell},0).$$
Now, we have 
$${\bf u}-{\bf v}=(2^{\delta+1}(qc_1-C_1a-D_1)-2^{\ell},\ldots,2^{\delta+1}(qc_n-C_na-D_n)-2^{\ell},-a),$$
and inequalities (\ref{inequality}) yield:
\begin{equation} \label{enumeration bound}
\|\mathbf{u}-\mathbf{v}\| < 2^{\ell}\sqrt{n+1}.
\end{equation}
Put $R = 2^{\ell}\sqrt{n+1}$. Then $\mathbf{u}\in B_{\mathbf{v}}(R)$.

Next, we compute a $LLL$-reduced
basis for $\mathcal{L}$, say $\mathcal{B} = \{\mathbf{b}_1,\ldots,\mathbf{b}_{n+1}\}$. This can be done in time ${\mathcal{O}}(n^6 (\log q)^3)$. By hypothesis (2) of Theorem,
 we have:
$$s(\mathcal{L}) > \frac{1}{2} \,(2^{\delta+1} q)^{\frac{n}{n+1}}.$$
Let  $\{\mathbf{b}_1^*,\ldots,\mathbf{b}_{n+1}^*\}$ the  Gram-Schmidt orthogonalisation of $\mathcal{B}$.
By   
 \cite[Theorem 6.66]{Hoffstein}, we get:
 $$4\|{\bf b}_i^*\|^2 \geq 2 \|{\bf b}_{i-1}^*\|^2 \geq \|{\bf b}_{i-1}\|^2 \geq s(L)^2$$
Thus, we obtain: 
\begin{equation} \label{GSO bound}
\frac{1}{4} \, (2^{ \delta+1}q)^{\frac{n}{n+1}} \leq
 \|\mathbf{b}_i^*\| \ \ \ (i=1,\ldots,n+1).
 \end{equation}
Next, using Kannan's enumeration algorithm \cite{Hanrot}, we compute
all  the elements of $B_{\mathbf{v}}(R)\cap \mathcal{L}$. 
Combining \cite[Theorem 5.1]{Hanrot2} with the inequality (\ref{GSO bound}),
we obtain that the bit complexity of the procedure is
$$(n\log q)^c\ 2^{\mathcal{O}(n)} \left(\frac{2^{\ell+2}}
{(2^{\delta+1}q)^{\frac{n}{n+1}}} \right)^{n+1} ,$$
where $c$ is a constant $>0$. 
 Then we check whether the last coefficient of
 $\mathbf{u} \in B_{\mathbf{v}}(R)\cap \mathcal{L} $ is the  minus of the secret
 key $-a\mod{q}$. Every such operation needs $\mathcal{O}((\log q)^4)$ bit operations
 \cite[Lemma 6.2, p.237]{Zheng}. If none of the elements of 
 $\mathbf{u} \in B_{\mathbf{v}}(R)\cap \mathcal{L} $ gives the secret key, then
 we repeat the procedure assuming that $k_1 = \min\{k_0,\ldots,k_n\}$, and we continue until 
 we found the secret key. 
 By  Lemma \ref{sphere}, we have:
$$|B_{\mathbf{v}}(R)\cap \mathcal{L} | < \left(
\frac{ 2^{\ell+2} \sqrt{n+1}}
{ (2^{\delta+1}q)^{\frac{n}{n+1}}} +1\right)^{n+1}.$$
Thus, the overall bit complexity of the computation of $a$ is
 $$\mathcal{O}\bigg(n(n\log q)^c\ 2^{\mathcal{O}(n)} \left(\frac{2^{\ell+2}}
{(2^{\delta+1}q)^{\frac{n}{n+1}}} \right)^{n+1} +n
 \left(
\frac{ 2^{\ell+2} \sqrt{n+1}}
{ (2^{\delta+1}q)^{\frac{n}{n+1}}} +1\right)^{n+1} (\log q)^4\bigg),$$
whence the result.

\section{The attack}

The proof of Theorems 1.1 yields the following attack: 

\begin{algor} {\rm
ATTACK-DSA\\
{\it Input:}  Messages $m_j $ $(j=0,\ldots,n)$   and 
$(r_j,s_j)$ their (EC)DSA signatures and  integers $\delta >0$  and 
 $0 \leq \delta_L\le \delta$ and the public key  $(p,q,g,R)$ 
 (resp. $(E,p,q,P,Q)$). \\
{\it Output:}  The  private key $a$.

\begin{enumerate}

\item  For $j=0,\ldots, n$ compute  $A_j =  -r_is_i^{-1} \bmod\ q$,
 $B_j = -h(m_j) s_j^{-1} \bmod\ q$.
 
 \item For $i=0,\ldots,n$,
 
 \begin{enumerate}
\item  For $j=1,\ldots,i$ compute  
$$C_{i,j} = (A_{j-1} -A_i) 2^{-\delta_L} \ \bmod\ q, \  \  D_{i,j} = (B_{j-1} -B_i) 2^{-\delta_L} \ \bmod\ q,$$
and for $j= i+1,\ldots,n$ compute
$$C_{i,j} = (A_{j} -A_i) 2^{-\delta_L} \ \bmod\ q, \ \  
D_{i,j} = (B_{j} -B_i) 2^{-\delta_L} \ \bmod\ q.$$

\item Consider the lattice $\mathcal{L}_i$ spanned by the rows of the matrix
\[
J_i =   \left( \begin{array}{ccccccc}
		2^{\delta+1}q & 0 & 0 & \ldots  & 0 & 0 \\
		0             & 2^{\delta+1} q & 0 & \ldots  & 0 & 0 \\
		0             & 0 & 2^{\delta+1} q & \ldots  & 0 & 0 \\
		\vdots        & \vdots & \vdots & \ddots & \vdots & \vdots \\
		0             &    0  & 0   &  \ldots & 2^{\delta+1} q  & 0\\
		2^{\delta+1}C_{i,1} & 2^{\delta+1}C_{i,2} & 2^{\delta+1}C_{i,3}  & \ldots & 2^{\delta+1}C_{i,n} & 1
\end{array}
\right)
\]
and compute a $LLL$-basis $\mathcal{B}_i$ for $\mathcal{L}_i$.

\item
Consider the vector
 $\mathbf{v}_i = (2^{\delta+1}D_{i,1}+2^{\ell},\ldots,2^{\delta+1}D_{i,n}+2^{\ell},0)$, and 
using Kannan's enumeration algorithm
 with basis $\mathcal{B}_i$, compute all
 $\mathbf{u}\in \mathcal{L}_i$ satisfying 
 $\|\mathbf{u}-\mathbf{v}_i\| < 2^{\ell}\sqrt{n+1}$. 
\item Check  whether the last coordinate of $\mathbf{u}$ say $u_{n+1}$ satisfies $g^{-u_{n+1}}\equiv R\pmod{q}$ (resp. $P(-u_{n+1}) = Q)$.
 If it is so, then return the secret key $-u_{n+1}\mod{q}=a$.
\end{enumerate}
\end{enumerate}
}
\end{algor}
\begin{remark}{\rm  For the Pseudocode of Kannan's Enumeration Algorithm, one
can see \cite[Section 5.1, Algorithm 10]{Hanrot2}.
}
\end{remark}
\begin{remark}{\rm
Supposing that condition (2) is satisfied, taking  $n$  quite small and $n\delta \geq \ell$, Theorem \ref{theorem} implies that the attack is polynomial
in $\ell$. Furthermore, if $s(L)$ is closed to the Gauss heuristic, then   the upper bound for the number of points of   $B_{\mathbf{v}}(R)\cap \mathcal{L}$ will be the smaller possible, and so 
it is expect that the attack will be quite efficient.
}
\end{remark}
{\bf Experiments.}
We conducted a thorough analysis of our experiments, and we compared our results with those presented by Gomez et al. \cite{Gomez}. Our findings indicate a significant improvement in almost all cases. Our experiments were conducted on a Linux machine with an i5-12400 CPU, using Sagemath 9.8 \cite{sage}. We made the assumption that we already knew the minimum ephemeral key. However, in the general case, where the minimum key is unknown, we would need to perform $n$ executions, where $n+1$ represents the number of signatures. This worst-case scenario would require multiplying the execution time of each experiment by $n$. Overall, our results demonstrate a notable improvement compared to the previous findings (see the Table below). Finally, we have successfully found the secret key even when the shared bits in the ephemeral keys are only 5. Remarkably, in this case, we only needed a minimum of 58 signatures. It is worth noting that in  \cite{Gomez}, no successful attack was provided for the specific scenario where $\delta=5.$

\begin{table}[h]
	\begin{center}
		\begin{tabular}{|c||c|c|}
			\hline
			$\delta:$ shared bits & signatures (\cite{Gomez}) & signatures (this paper)   \\ \hline
			$5$ &   $?$ & $58$  \\ \hline
			$6$ &   $\approx 50$ & $40$  \\ \hline
			$8$ &  $\approx 27$   & $25$ \\   \hline
			$10 $ & $\approx 20$  & $18$ \\   \hline
			$12$ &  $\approx 17$  & $14$ \\ \hline
			$14 $ &  $\approx 14$ & $12$ \\ \hline
			$16 $ &  $\approx 12$ & $11$ \\   \hline
			$18 $ &  $\approx 11$ & $9$ \\   \hline
			$20 $ &  $\approx 10$ & $8$ \\   \hline
		\end{tabular}
	\end{center}
	\caption{\footnotesize We considered $\ell=160$ and $R=2^{\ell}\sqrt{n+1}.$ We found the private key in every experiment we executed. Instead of LLL we used BKZ with block size $8.$ For each row, we conducted 10 random experiments, and the results in the third column were computed on average in under two minutes and thirty seconds.  For the cases, $\delta\in \{22,24,26,28,30\}$ we get the same number of signatures as in \cite{Gomez}. See \url{https://github.com/drazioti/dsa/}.	
}
\label{Tab:Table1_determenistic}
\end{table}

\section{Conclusion}
 
Attacks based on the malicious modification of memory registers
is a  topic of high importance, since it 
 may affect the randomness of the secret parameters by forcing a limited number of bits to a certain value, which can be unknown to the attacker.
 In this context,  we developed a deterministic attack on the DSA schemes,  
providing that several signatures are such that the corresponding 
ephemeral keys share a number of bits without knowing their value.

Our attack is deterministic, meaning that it  always produces a result for a given input every time. However, it is important to note that while the attack is deterministic, it may not always be practical to execute. Deterministic attacks on the (EC)DSA are relatively rare, as they typically rely on heuristic assumptions.  

While deterministic attacks on (EC)DSA, are rare, our attack demonstrates practical feasibility in specific scenarios, surpassing previous results, (see Table \ref{Tab:Table1_determenistic}). However, it is important to note that the practicality and effectiveness of our attack may vary depending on the specific choice of $(\delta,n)$.

\ \\
{\bf Acknowledgement}\\
The  author, Marios Adamoudis is co-financed by Greece
 and the European Union (European Social Fund-ESF) through the Operational
 Programme ”Human Resources Development, Education and Lifelong Learning”
 in the context of the Act ”Enhancing Human Resources Research Potential by
 undertaking a Doctoral Research” Sub-action 2: IKY Scholarship Programme for
 PhD candidates in the Greek Universities.

\begin{figure}[h] 
\centering{\includegraphics[scale = 0.4]{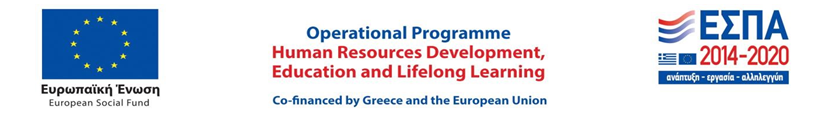}
}
\end{figure}


{\small
\begin{thebibliography}{99}


\bibitem{marios} M. Adamoudis, K. A. Draziotis and D. Poulakis, Enhancing a DSA attack, CAI 2019, p. 13-25.  LNCS {\bf 11545}, Springer 2019.

\bibitem{Aranha} 
D. F. Aranha, F. R. Novaes, Akira Takahashi, M. Tibouchi,
and Y. Yarom. LadderLeak: Breaking ECDSA with less than one bit of
nonce leakage. In Jay Ligatti, Xinming Ou, Jonathan Katz, and Giovanni
Vigna, editors, ACM CCS 2020, pages 225-242. ACM Press, November 2020.

\bibitem{Bellare}  M. Bellare, S. Goldwasser and Micciancio,
``Pseudo-random" number generation within cryptographic
algorithms: the DSS case. In {\it Proc. of Crypto '97,} LNCS {\bf  1294}
IACR, Palo Alto, CA. Springer-Verlag, Berlin 1997.

\bibitem{Blake}  I. F. Blake and T. Garefalakis,
 On the security of the digital signature algorithm.
 {\it Des. Codes Cryptogr.,} 26,  no. 1-3 (2002), 87-96.
 
 \bibitem{Bleichenbacher}
 D. Bleichenbacher. On the generation of one-time keys in DL signature
schemes. In Presentation at IEEE P1363 working group meeting, page 81,
2000.

 \bibitem{Draziotis} K. A. Draziotis and D. Poulakis, Lattice attacks on DSA schemes based on Lagrange's algorithm.
 5th international Conference on Algebraic Informatics, CAI  2013.  Berlin: Springer. LNCS {\bf 8080}, 119-131 (2013).

\bibitem{Draziotis2}  K. A. Draziotis, (EC)DSA lattice attacks based on Coppersmith's method, Information Processing Letters 116(8), Elsevier (2016), Pages 541-545.

\bibitem{ElGamal} T. ElGamal, A public key cryptosystem and a signature scheme 
based on discrete logarithm, {\it IEEE
Transactions on Information Theory,} 31 (1985), 469-472.


\bibitem{fips} FIPS PUB 186-3, Federal Information Processing Standards 
Publication, Digital Signature Standard (DSS).

\bibitem{Faugere} J. -L. Faug\`{e}re, C. Goyet, and G. Renault, { Attacking 
(EC)DSA Given Only  an Implicit Hint, Selected Area of Cryptography,} LNCS {\bf 7707}, p. 252--274, Springer-Verlag, Berlin - Heidelberg 2013.

\bibitem{Gomez} Ana I. Gomez, D. Gomez-Perez, and G. Renault,  A probabilistic analysis on a lattice attack against DSA. {\it Des. Codes Cryptogr.} 87, 2469-2488 (2019). 

\bibitem{Hanrot} G. Hanrot and D. Stehl\'{e}, Improved analysis of kannan’s shortest lattice vector algorithm. In
Proceedings of Crypto, LNCS {\bf 4622}, 170-186. Springer,  2007.

\bibitem{Hanrot2} G. Hanrot, X. Pujol and  D. Stehl\'{e}, 
Algorithms for the shortest and closest lattice vector problems. 
Chee, Yeow Meng (ed.) et al., Coding and cryptology. Third international workshop, IWCC 2011, Qingdao, China, May 30 – June 3, 2011. Proceedings. Berlin: Springer. Lecture Notes in Computer Science 6639, 159-190 (2011). 


\bibitem{Hoffstein} J. 
Hoffstein, J. Pipher, H. H.  Silverman, {\it 
An introduction to mathematical cryptography.} 2nd ed. 
Undergraduate Texts in Mathematics. New York, NY: Springer 2014. 

\bibitem{Howgrave} N. A. Howgrave-Graham and N. P. Smart, Lattice
Attacks on Digital Signature Schemes,  {\it Des. Codes Cryptogr.}
23 (2001) 283-290. 

\bibitem{Johnson} D. Johnson, A. J. Menezes and S. A. Vastone, The
elliptic curve digital signature algorithm (ECDSA), {\it Intern.
J. of Information Security,} 1 (2001) 36-63.
 
 
\bibitem{Koblitz} N. Koblitz, A. J. Menezes and S. A. Vastone,
The state of elliptic curve cryptography, {\it Des. Codes
Cryptogr.} 19 (2000), 173-193.

\bibitem{Koblitz2} N. Koblitz and  A. J. Menezes, A survey of Public-Key 
Cryptosystems,
{\it SIAM REVIEW,} 46, No. 4 (2004), 599-634.

\bibitem{Leadbitter} P.J. Leadbitter, D. Page, N.P. Smart. Attacking DSA Under a Repeated Bits Assumption. In: Joye, M., Quisquater, JJ. (eds) Cryptographic 
Hardware and Embedded Systems - CHES 2004. CHES 2004. Lecture Notes in Computer 
Science, vol 3156, (2004) 428-440. Springer, Berlin, Heidelberg.


\bibitem{Lenstra} A. K. Lenstra, H. W. Lenstra Jr., and L. Lov\'{a}sz, Factoring 
polynomials
with rational coefficients, {\it Math. Ann.}, 261 (1982), 513-534.

\bibitem{Malikiosis} R.-D. Malikiosis, 
Lattice-point enumerators of ellipsoids, {\it 
Combinatorica} 33, No. 6 (2013) 733-744.


\bibitem{Menezes} A. J. Menezes, P. C. van Oorschot and S. A.
Vanstone, {\it Handbook of Applied Cryptography,} CRC Press, Boca
Raton, Florida, 1997.

\bibitem{Micciancio} D. Micciancio and P. Voulgaris. A deterministic single 
exponential time algorithm
for most lattice problems based on Voronoi cell computations. {\it In Proc. of 
STOC, ACM}, (2010) pages 351-358.


\bibitem{Mulder1}
E. De Mulder, M. Hutter, M. E. Marson, and P. Pearson. Using Bleichenbacher s solution
to the Hidden Number Problem to attack nonce leaks in 384-bit ECDSA. In {\it Cryptographic Hardware
and Embedded Systems-CHES 2013,} 435-452. Springer, 2013.

\bibitem{Mulder2}
E. De Mulder, M. Hutter, M. E. Marson, and P. Pearson. Using Bleichenbacher's solution
to the hidden number problem to attack nonce leaks in 384-bit ecdsa: extended version. {\it Journal of
Cryptographic Engineering,} 4(1):33-45, 2014.


\bibitem{National} National Institute of Standards and Technology
(NIST). {\it FIPS Publication {\rm 186}: Digital Signature
Standard.} May 1994.

\bibitem{Nguyen} P. Nguyen and I. E. Shparlinski, The Insecurity
of the Digital Signature Algorithm with Partially Known Nonces,
{\it J. Cryptology,} 15 (2002), 151-176.

\bibitem{Nguyen2} P. Nguyen and I. E. Shparlinski,
The Insecurity of the Elliptic Curve Digital Signature Algorithm
with Partially Known Nonces,  {\it Des. Codes Cryptogr.} 30,
(2003), 201-217.

\bibitem{Poulakis} D. Poulakis, Some Lattice Attacks on DSA and ECDSA, {\it Applicable Algebra in Engineering, Communication and Computing,}
22, (2011), 347-358.

\bibitem{Poulakis1} D. Poulakis, New lattice attacks on DSA schemes, 
 {\it J. Math. Cryptol.}  10 (2) (2016), 135–144.

\bibitem{sage} Sage Mathematics Software, The Sage Development Team. \url{http://www.sagemath.org}.

 \bibitem{Sun}
 C. Sun, T. Espitau, M. Tibouchi, and M. Abe, Guessing Bits: Improved 
 Lattice Attacks on (EC)DSA with Nonce Leakage, 
{\it IACR Transactions on Cryptographic Hardware and Embedded Systems,}
ISSN 2569-2925, Vol. 2022, No. 1, pp. 391-413.

\bibitem{Zheng} Z. Zheng, {\it Modern Cryptography, Volume 1,} 
Springer 2021. 

\end{thebibliography}
}
\end{document}